\documentclass[12pt]{amsart}
\usepackage[T1]{fontenc}
\usepackage{amsmath, amsfonts, amsthm, amssymb}

\allowdisplaybreaks[3]
\usepackage{graphicx}

\usepackage[margin=1in]{geometry}
\usepackage{color}
\usepackage{amssymb}
\usepackage{amsmath}
\usepackage{amsthm}
\usepackage{url}
\usepackage{framed}
\usepackage{enumitem}
\usepackage{multirow}

\usepackage[svgnames]{xcolor}
\usepackage{tikz}
\usepackage{pgfplots}
\pgfplotsset{compat=1.3}

\newtheorem{theorem}{Theorem}[section]
\newtheorem{proposition}[theorem]{Proposition}
\newtheorem{lemma}[theorem]{Lemma}
\newtheorem{corollary}[theorem]{Corollary}

\newtheorem{conjecture}[theorem]{Conjecture}

\theoremstyle{definition}
\newtheorem{definition}[theorem]{Definition}
\newtheorem{remark}[theorem]{Remark}
\newtheorem{example}[theorem]{Example}

\newcommand{\Pc}{\mathcal{P}}

\newcommand{\CHI}{\hbox{\raise .4ex \hbox{$\chi$}}}

\newcommand{\Fc}{{\mathcal{F}}}

\newcommand{\R}{{\mathbb{R}}}
\newcommand{\Sb}{\mathbb{S}}

\newcommand{\Sd}{\mathbb{S}^{d-1}}
\newcommand{\nm}[1]{\|{#1}\|}
\newcommand{\ip}[2]{\langle#1,#2\rangle}
\newcommand{\fp}{\text{FP}}
\newcommand{\pfp}{\text{PFP}}
\definecolor{darkgreen}{RGB}{60,115,60}

\newcommand{\qt}[1]{\quad \text{#1} \quad}

\DeclareMathOperator{\Span}{span}

\begin{document}

\title{Universal optimal configurations for the  $p$-frame potentials}

\author{X.~Chen}
\address{Department of Mathematical Sciences\\
New Mexico State University\\
Las Cruces, NM 88003}
\email{xchen@nmsu.edu}

\author{V.~Gonzales}
\address{Department of Mathematics\\
University of Maryland\\
College Park\\
MD 20742}
\email{victor.gonzalez578@yahoo.com}

\author{E.~Goodman}
\address{Department of Mathematics\\
University of Pennsylvania\\
David Rittenhouse Lab\\
209 South 33rd Street\\
Philadelphia, PA 19104}
\email{ericgood@sas.upenn.edu}

\author{S.~Kang}
\address{Department of Mathematics and Norbert Wiener Center\\
University of Maryland\\
College Park\\
MD 20742}
\email{kangsj@math.umd.edu}

\author{K.~A.~Okoudjou}
\address{Department of Mathematics and Norbert Wiener Center\\
University of Maryland\\
College Park\\
MD 20742}
\email{kasso@math.umd.edu}

\maketitle

\begin{abstract}
Given $d, N\geq 2$ and $p\in (0, \infty]$ we consider a family of functionals, the $p$-frame potentials FP$_{p, N, d}$, defined  on the set of all collections of $N$ unit-norm vectors in $\mathbb R^d$.  For the special case $p=2$ and $p=\infty$, both the minima and the minimizers of these potentials  have been thoroughly investigated.  In this paper, we investigate the minimizers of the functionals FP$_{p, N, d}$, by first establishing  some general properties of their minima. Thereafter, we focus on the special case $d=2$, for which, surprisingly, not much is known. One of our main results establishes the unique minimizer for big enough $p$. Moreover, this minimizer is universal in the sense that it minimizes a large range of energy functions that includes the $p$-frame potential. We conclude the paper by reporting some numerical experiments for the case  $d\geq 3$, $N=d+1$,  $p\in (0, 2)$. These experiments lead to some conjectures that we pose. 
\end{abstract}

\section{Introduction} 
 A set of vectors $X=\{x_k\}_{k=1}^{N}
\subseteq \mathbb{\R}^d$ is a \textit{frame} for $\mathbb{R}^d$ if there exist $0<A\leq B < \infty$ such that 
 \begin{equation}\label{eq:frineq}
   \quad A\nm{x}{}^{2} \leq \sum_{k=1}^{N} | \ip{x}{x_k}|^{2}
   \leq B \nm{x}{}^{2} \quad\text{for all }x \in \mathbb{R}^d,
\end{equation}
where $\|\cdot\|$ denotes the Euclidean norm.
If, in addition, each $x_k$ is unit-norm, we
say that $X$ is a {\it unit-norm frame}. $X$ is called \emph{tight} if $A=B$.
A tight unit-norm frame is called a {\it finite unit-norm tight frame (FUNTF)}.  One attractive feature  of FUNTFs   is the fact that they can be used to decompose and reconstruct any vector $x$ via the following formula:
\begin{equation}\label{funtfexpan}
x =\frac{d}{N} \sum_{k=1}^{N}\ip{x}{x_{k}}x_{k}.
\end{equation} 
Frames in general,  and FUNTFs in particular, are routinely used in many applications, especially in signal processing. For more on the theory and the applications of frames we refer to   \cite{CasazzaFramesChapter, KovChe1, KovChe2, OkoudjouFiniteFrame}.  

A  frame $X$ is said to be {\it equiangular} if there exists $c>0$ such that 
$$\Big|\Big\langle  \frac{x_k}{\|x_k\|}, \frac{x_l}{\|x_l\|}\Big\rangle \Big|=c\quad \text{ for all }k\neq l.$$ 
If in addition $X$ is  tight, then $X$ is called an \emph{equiangular tight frame (ETF)}. It follows from \cite[Proposition 1.2]{BodPauTom2009} that  the vectors   of an ETF have  necessarily  equal norm. Consequently,  and  without loss of generality, all ETFs considered in the sequel  will be unit-norm frames, i.e., FUNTFs.

Let $S(N,d)$ be the collection of all sets of   $N$ unit-norm vectors.  For any $p\in (0, \infty]$, the \emph{$p$-frame potential} of $X=\{x_k\}_{k=1}^{N}\in S(N,d)$ is defined as 
\begin{equation}\label{eq:pfrpot}
\fp_{p,N,d}(X):=\left\{\begin{array}{ll}\displaystyle\sum_{k=1}^{N}\sum_{l\neq k}^N|\ip{x_k}{x_{\ell}}|^p,\, &\textrm{when}\, p<\infty\\
\displaystyle\max_{ k\neq \ell} |\ip{x_{k}}{x_{\ell}}|,&\text{when }p=\infty.\\
\end{array}\right.
\end{equation} 
The definition of the $p$-frame potential above differs from the  one given in  \cite{EhlOko2012} as \eqref{eq:pfrpot} excludes self inner products. As will be seen in Section~\ref{sec:basic}, the present definition will allow us to state our results in a more concise manner.  The subscripts $N, d$ are a little redundant since they are suggested by the input $X$, but they will come handy when we want to emphasize the dimension or the number of points. We are interested in finding the infimum of the $p$-frame potential among all $N$-point configurations in $S(N,d)$. It is a standard argument to show that this infimum can be achieved due to the compactness of the sphere and the continuity of the function, so we can replace infimum by minimum and define
\begin{equation}\label{equ:min}
\Fc_{p,N,d}:=\min_{X\in S(N,d)}\fp_{p,N,d}(X).
\end{equation}

 In situations when $N,d$ are both fixed, we will simply use $\Fc_p$ for $\Fc_{p,N,d}$, and $\fp_p$ for $\fp_{p,N,d}$. Similarly we use the notations $\Fc_N,\fp_N$ if  $p$ and $d$ are fixed.
Any minimizer of \eqref{equ:min} will be called an \emph{optimal configuration} of the $p$-frame potential. We observe that if $X^*=\{x_1^*,\cdots,x_N^*\}$ is optimal, then with any orthogonal matrix $U$, any permutation $\pi$, and any $s_i\in\{1,-1\}$,
$$\{s_1 Ux_{\pi_1}^*,\cdots, s_N Ux_{\pi_N}^*\}$$ is optimal too. In other words, the optimal configuration is an equivalence class with respect to orthogonal transformations, permutations or sign switches.
So when we say an optimal configuration is unique, we mean that it is  unique up to this equivalence relation.

Note that in the definition of the frame potential, $X$ does not necessarily need to be a frame of $\R^d$, but we will show in Proposition \ref{prop:general} that the minimizers of the $p$-frame potential must be a frame, as expected. Therefore problem \eqref{equ:min} remains the same if we had restricted $X$ to be a unit-norm frame with $N$ frame vectors.

The name ``frame potential'' originates from the special case $p=2$, 
\begin{equation}\label{eq:frpot}
\fp_{2,N,d}(X) = \sum_{k=1}^{N} \sum_{l\neq k}^{N} |\ip{x_{k}}{x_{l}}|^{2}
\end{equation}
which was studied by Benedetto and Fickus~\cite{BenFic2003}.
 They  proved that $X^*$ is an optimal configuration of $\fp_{2,N,d}(X)$ if and only if $X^*=\{x_k^*\}_{k=1}^N$ is a FUNTF.

Another important special case is $p=\infty$. In this case, the quantity
\begin{equation}
c(X):=\fp_{\infty, N, d}(X)=\max_{ k\neq \ell} |\ip{x_{k}}{x_{\ell}}|
\end{equation}
is also called the \emph{coherence} of $X=\{x_k\}_{k=1}^N\in S(N,d)$, and its minimizers are called \emph{Grassmanian frames} \cite{BenKeb2008, Bukcox18, HeaStr2003, welch74}.  The following Welch bound~\cite{welch74} is well known:
\begin{equation}\label{welch}
\fp_{\infty, N, d}(X)\geq \sqrt{\frac{N-d}{d(N-1)}},
\end{equation}
and the equality  in~\eqref{welch} holds if and only if $X=\{x_k\}_{k=1}^N$ is an ETF,
which is only possible if $\displaystyle N\leq \frac{d(d+1)}{2}$. 
The coherence minimization problem corresponds to   $p=\infty$ because it appears to be the limiting case when $p$ grows to infinity; see Proposition \ref{prop:infty}.

 When  $p$ is an even integer, the minimizers of $\fp_{p, N, d}$ have long been investigated in the setting of spherical designs, see \cite{EhlOko2012, venkov2001reseaux}. 
 A set of $N$ points $X\subset \Sd$ (the unit sphere in $\R^d$) is called a spherical $t-$design if for every homogeneous polynomial $h$ of degree $t$ or less,
$$\int_{\Sd} h(\xi)d\sigma(\xi)=\frac{1}{N}\sum_{x\in X}h(x),$$
where $\sigma$ is the normalized surface measure on $\Sd$. For example, a spherical $1$-design is a set of points whose  center of mass is at the origin. More generally, as shown in \cite{EhlOko2012} or \cite[Theorem 8.1]{venkov2001reseaux}, if $p$ is an even integer and $X\in S(N,d)$ is symmetric, that is $X=-X$, then  
\begin{equation}\label{equ:sph}\fp_{p,N,d}(X)\geq N^2\frac{1\cdot3\cdot5\cdots (p-1)}{d(d+2)\cdots(d+p-2)}-N,
\end{equation}
and equality holds if and only if $X$ is a spherical $p$-design.

Optimal configurations of \eqref{equ:min} are often not symmetric since $x_i$ and $-x_i$ are considered the same points as far as frame potential is concerned. However, we can still use \eqref{equ:sph} by symmetrizing a frame. Given $X=\{x_i\}_{i=1}^N$ such that its coherence $c(X)<1$ (i.e. no repeated vectors or opposite vectors),  we let   
$$X^{sym}:=\{x_i\}_{i=1}^N\cup\{-x_i\}_{i=1}^N\in S(2N,d).$$
Some straightforward computations result in 
\begin{equation}\label{equ:sym}
\fp_{p, 2N}(X^{sym})=4\fp_{p,N}(X)+2N
\end{equation} 
which combined with \eqref{equ:sph}, can be used to prove

\begin{proposition}\label{thm:spherical}
Let $p$ be an even integer, then
$$\fp_{p,N,d}(X)=\frac{1}{4}(\fp_{p, 2N}(X^{sym})-2N)\geq N^2\frac{1\cdot3\cdot5\cdots (p-1)}{d(d+2)\cdots(d+p-2)}-N,$$
and equality holds if and only if $X^{sym}$ is a spherical $p$-design.
\end{proposition}

Not only is Proposition \ref{thm:spherical} limited to even $p$'s, but it is also not trivial to find spherical $t$-designs for large $t$.
More generally, and to the best of our knowledge, little is known about the complete solutions to \eqref{equ:min} even in the simplest case  $d=2$. When $N=3$,  a solution is given in \cite{EhlOko2012} for all positive $p$.  See also \cite{BilMat18, Oktay} for related results.
For any $N$ and $p=\infty$, it is shown in ~\cite{BK06} that the Grassmannian frame is 
\begin{equation}\label{equ:half}
X_N^{(h)}=\{e^{i\cdot 0}, e^{i\frac{\pi}{N}},e^{i\frac{2\pi}{N}},\cdots, e^{i\frac{(N-1)\pi}{N}}\},
\end{equation} 
which can be viewed as $N$  equally spaced points on the half circle.
The main result of this paper establishes that the unique optimal configuration when $d=2$,  $N\geq4$, and   $p>4\lfloor \frac{N}{2}\rfloor-2$ is $X_N^{(h)}$, where $\lfloor c \rfloor$ is the largest integer that does not exceed $c$. Moreover for $N=4$, our result is sharper as we prove this is the case for $p>2$. 
Such a result is expected since optimal configurations for large $p$ are approaching the Grassmannian frame. Moreover, we are able to show that $X_N^{(h)}$ is the optimal configuration for a big class of kernel functions. See  Theorem \ref{thm:both}. The phenomenon that a given configuration is the optimal configuration for a large range of functions is what we call universal. Such a name stems from the work \cite{CohKum2007}. In addition to these results, we present numerical results for all other values of $p$ and $N$ when $d=2$. Finally, we also consider the special case of $N=d+1$ and $d\geq 3$ and state a conjecture regarding the function $\Fc_p$ for $p\in (0,2]$.
Based on the results of the present paper, Table \ref{tab:results} gives the state of affairs concerning  the solutions of~\eqref{equ:min} and is an invitation to initiate a broader discussion on the problem.

The rest of the paper is organized as follows. Section \ref{sec:basic} states some basic results of the $p$-frame potential including some asymptotic results as $N\rightarrow\infty$. Section \ref{sec3} presents the results for $d=2$. Section \ref{sec:d+1} presents conjectures and numerical results for the case  $N=d+1$. Throughout the paper, we will use $[m:n]$ for the index set $\{m, m+1, \cdots, n\}$.

\renewcommand*{\arraystretch}{1.5}
\begin{table}[ht]
\caption{Optimal configurations for the $p$-frame potential}\label{tab:results}
\begin{tabular}{l|l|l}
            & $\R^2$ & $\R^d$ \\
            \hline
 $p\in(0, \frac{\ln3}{\ln2})$ &$N=3$: ONB+ \cite{EhlOko2012} & $N=d+1$: ONB+ \cite{Glaz19}\\
 \hline
 $p\in(\frac{\ln3}{\ln2},2)$ &$N=3$: ETF \cite{EhlOko2012} & $N=d+1$: see Conjecture \ref{conj} \\\hline
 $p\in(0,2)$ & $N=2k$: $k$ copies of ONB \cite{EhlOko2012}    &$N=kd$: $k$ copies of ONB \cite{EhlOko2012}\\
 \hline
 $p=2$ & \multicolumn{2}{c}{FUNTF \cite{BenFic2003}}\\
 \hline
 %$p\in(2,\infty)$ & \multicolumn{2}{c}{ETF if exists \cite{EhlOko2012, HeaStr2003}}\\
 %\hline
 %$p\in(2, 4\lfloor \frac{N}{2}\rfloor-2)$ & open for $N\geq 5$ &\multirow{4}{*}{ETF if exists \cite{EhlOko2012, HeaStr2003}}\\
% \hline
 $p\in(4\lfloor \frac{N}{2}\rfloor-2, \infty)$ & $N\geq5: X_N^{(h)}$ (Theorem \ref{thm:p})&\multirow{3}{*}{ETF if exists \cite{EhlOko2012, HeaStr2003}} \\
 \cline{1-2}
 $p=\infty$ &Any $N$: $X_N^{(h)}$ \cite{BK06}  \\
 \cline{1-2}
 $p\in(2,\infty)$ & $N=4$: $X_4^{(h)}$ (Theorem \ref{thm:p}) \\
 \hline
\end{tabular}

\vspace{0.1in}

ONB+  refers to an orthonormal basis with a repeated vector. See Definition \ref{def:onbp}(a).

\end{table}

\renewcommand*{\arraystretch}{1}

\section{Some basic results}\label{sec:basic}

Intuitively, minimizing the frame potential amounts  to promoting big angles among vectors. Consequently,  it is expected that the optimal configurations will be at least a frame whose vectors are  reasonably spread out in the sphere. If $X$ is not a frame, then one can always find a vector $e$ that is orthogonal to $X$, and replacing any vector in $X$ by $e$  won't increase the frame potential. In other words, it is trivial to show that problem \eqref{equ:min} might as well be restricted to frames. The following result shows something stronger, that is, it excludes the possibility that a minimizer doesn't span $\R^d$. 

\begin{proposition}\label{prop:general}
For $p\in(0,\infty]$, any optimal configuration of \eqref{equ:min} is a frame of $\R^d$.
\end{proposition}
\begin{proof} 
We first consider the case $p\in (0, \infty)$. 
Suppose not, and say $X^{*}=\{x_k^*\}_{k=1}^N\subset S^{d-1}$ is a minimizer so that $\Span X^*$ is a strict subset of $\R^d$. Because there are $N \ge d$ vectors, it is possible to select two indices $k_1$ and $k_2$ such that $ |\ip{x^*_{k_1}}{x^*_{k_2}}| > 0$. Finally, select any unit-norm vector $x_0 \in (\Span X^*)^\perp$ and replace $x^*_{k_1}$ with $x_0$; i.e., define $Y = \{ x^*_k \}_{k \neq k_{1}} \cup \{ x_0 \}$. A direct computation shows that $\fp_{p, M, N}(Y) < \fp_{p, M, N}(X^*)$.

Now consider the case  $p=\infty$ and let $X^{*}=\{x_k^*\}_{k=1}^N\subset S^{d-1}$ be a minimizer  of $\fp_{\infty, N, d}$. Suppose that the dimension of $\Span(X^*) \leq d-1$.  Choose a unit vector $e \in (\Span X^*)^\perp$. There could be multiple pairs of vectors that achieve the maximal inner product $\Fc_{\infty}=\fp_{\infty, N,d}(X^*)$. Without loss of generality, we assume these vectors are among  the first $K$ vectors, that is, 
\begin{equation}\label{equ:K}
|\langle x_i^*, x_j^*\rangle|<\Fc_{\infty}, \quad\text{ if either $i$ or $j$ does not belong to} \, [1:K], i\neq j.
\end{equation}
 We will construct $Y=\{y_k\}_{k=1}^K\cup\{x_k^*\}_{k=K+1}^N$ that will have smaller coherence.

For $i=1,2, \cdots, K$, let $y_i=\sqrt{1-\epsilon_i}x_i^*+\sqrt{\epsilon_i}e$, where  $0<\epsilon_i<1$. Define $$f(a,b):=\frac{\sqrt{a}\sqrt{b}}{1-\sqrt{1-a}\sqrt{1-b}}\text{ on } (0,1]\times(0,1].$$ If we choose $\epsilon_i, \epsilon_j$ such that 
\begin{equation}\label{equ:mu}
f(\epsilon_i, \epsilon_j)=\frac{\sqrt{\epsilon_i}\sqrt{\epsilon_j}}{1-\sqrt{1-\epsilon_i}\sqrt{1-\epsilon_j}}<\Fc_{\infty}, 
\end{equation}
then 
\begin{equation}\label{equ:y}
|\langle y_i, y_j\rangle|=|\sqrt{1-\epsilon_i}\sqrt{1-\epsilon_j}\langle x_i^*,x_j^*\rangle+\sqrt{\epsilon_i}\sqrt{\epsilon_j}|\leq\sqrt{1-\epsilon_i}\sqrt{1-\epsilon_j}\Fc_{\infty}+\sqrt{\epsilon_i}\sqrt{\epsilon_j}<\Fc_{\infty}.
\end{equation}

We will pick $\epsilon_i$ iteratively to satisfy \eqref{equ:mu}:\newline
\noindent Step 1: pick $0<\epsilon_1<1$ arbitrarily.\newline 
%\noindent Step 2: pick $\epsilon_2>0$ such that $f(\epsilon_1, \epsilon_2)<\Fc_\infty$. This is possible because $\lim_{\epsilon\rightarrow0}f(\epsilon_2, \epsilon)=0$.\newline
\noindent Step $i$: given  $\epsilon_1, \cdots, \epsilon_{i-1}$, pick $\epsilon_i>0$ such that $f(\epsilon_j, \epsilon_i)<\Fc_\infty, \text{ for all }j=1,\cdots, i-1$. This is possible because $\lim_{\epsilon\rightarrow0}f(\epsilon_j, \epsilon)=0$ for all $j\leq i-1$.

 For convenience, let $y_k=x_k^*$ for $k=K+1, \cdots, N$. The new frame $Y=\{y_k\}_{k=1}^{K}$  has a smaller coherence because for any pair $i,j$, if $i,j\in [1:K]$, then $|\langle y_i, y_j\rangle|<\Fc_\infty$ by \eqref{equ:y}; if $i,j\in [K+1:N]$, then $|\langle y_i, y_j\rangle|=|\langle x_i^*, x_j^*\rangle|<\Fc_\infty$ by \eqref{equ:K}; if $i\in[1:K], j\in[K+1:N]$, then $|\langle y_i, y_j\rangle|=|\langle \sqrt{1-\epsilon_i}x_i^*+\sqrt{\epsilon_i}e, x_j^*\rangle|=\sqrt{1-\epsilon_i}|\langle x_i^*, x_j^*\rangle|<\Fc_\infty$.

This is a contradiction, so the optimal configuration must be a frame.
\end{proof}

Now we establish the relationship between large $p$ and $p=\infty$. 
\begin{proposition}\label{prop:infty}
$\lim_{p\rightarrow\infty}\Fc_p^{1/p}=\Fc_{\infty}$. Moreover, if $X^{(p)}$ is an optimal configuration for \eqref{equ:min}  when $p<\infty$ and $X$ is a cluster point of the set $\{X^{(p)}\}_{p>0}$, then $X$ optimizes the coherence as $\displaystyle X=\arg\min_{Y\in S(N,d)} c(Y)$. 
\end{proposition}
\begin{proof}
On one hand, we have
\begin{equation}\label{equ:1}
\Fc_p^{1/p}=\left(\sum_{i\neq j}|\langle x_i^{(p)}, x_j^{(p)}\rangle|^p\right)^{1/p}\geq c(X^{(p)})\geq\Fc_{\infty}.
\end{equation}
On the other hand, 
\begin{equation}\label{equ:2}
\Fc_p^{1/p}\leq\left(\sum_{i\neq j}|\langle x_i^{(\infty)}, x_j^{(\infty)}\rangle|^p\right)^{1/p}\leq \left(\sum_{i\neq j}\Fc_{\infty}^p\right)^{1/p}=\Fc_{\infty}[N(N-1)]^{1/p}.
\end{equation}
Taking the limit of both inequalities gives us the desired limit.

For the second part of the proposition, let $X=\lim_{k\rightarrow\infty}X^{(p_k)}$ where $p_k\rightarrow\infty$ as $k\rightarrow\infty$. Then by \eqref{equ:1} and \eqref{equ:2},
$$c(X^{(p_k)})\leq\Fc_{p_k}^{1/p_k}\leq\Fc_{\infty}[N(N-1)]^{1/p_k}.$$
Letting $k\rightarrow\infty$, by continuity of the coherence, we get $c(X)\leq\Fc_{\infty}$ which forces $c(X)=\Fc_{\infty}$.
\end{proof}

Next, we establish a continuity result of $\Fc_{p}$.

\begin{proposition}\label{prop:cont}
The minimal frame potential $\Fc_p$ is a  continuous and non-increasing function of $p\in(0,\infty)$. 
\end{proposition}
\begin{proof}
We first prove that the function is non-increasing. Letting $p>q$, for any $X\in S(N,d)$, 
$$\fp_q(X)\geq\fp_p(X)\geq\Fc_p,$$
so $\Fc_q=\fp_q(X^{(p)})\geq\Fc_p$.

For continuity, we have
\begin{equation*}\label{equ:frac}
\sum_{i\neq j,|\langle x_i, x_j\rangle|\neq0} |\langle x_i, x_j\rangle|^{q}\ln|\langle x_i, x_j\rangle|\leq \frac{\fp_p(X)-\fp_q(X)}{p-q},
\end{equation*}
which comes from applying the inequality $\displaystyle a^{q} \ln a\leq\frac{a^p-a^q}{p-q}$ for $0<q<p, a>0$ to every nonzero term in the frame potential. 

So
\begin{align*}
0&\leq \frac{\Fc_q-\Fc_p}{p-q}\\
&=\frac{\Fc_q-\fp_p(X^{(p)})}{p-q}\\
&\leq\frac{\fp_q(X^{(p)})-\fp_p(X^{(p)})}{p-q}\\
&\leq \sum_{i\neq j,|\langle x_i^{(p)}, x_j^{(p)}\rangle|\neq0} |\langle x_i^{(p)}, x_j^{(p)}\rangle|^{q}\ln\frac{1}{|\langle x_i^{(p)}, x_j^{(p)}\rangle|}\leq\sum_{i\neq j,|\langle x_i^{(p)}, x_j^{(p)}\rangle|\neq0}\ln\frac{1}{|\langle x_i^{(p)}, x_j^{(p)}\rangle|}:=C_p.
\end{align*}
Therefore $0\leq\Fc_q-\Fc_p\leq(p-q)C_p$, which implies the continuity of $\Fc$. 
\end{proof}

Next, for fixed $p,d$, we consider the asymptotics of $\Fc_{p, N, d}$ as the number of points $N$ grows. In particular, we show that $\Fc_N\sim N^2$, see Proposition~\ref{prop:asymp}. We note that this behavior was numerically  observed in  \cite{AvGoDuSt}.  We begin by establishing some preliminary results.

\begin{lemma}
Given $d\geq 2$, and  $p\in(0,\infty)$, the sequence $\left\{\frac{\Fc_{p,N}}{N(N-1)}\right\}_{N\geq d+1}$ is a non-decreasing sequence.
\end{lemma}

\begin{proof}
Let $X^{(N)}=\{x_i^{(N)}\}_{i=1}^N$ be an optimal configuration for $\fp_{p,N,d}$. For each $k\in[1:N]$,
\begin{align}\label{equ:inc}
\Fc_N=\fp(X^{(N)})=\fp\left(X^{(N)}\backslash\{x_k^{(N)}\}\right)+2\sum_{j\neq k}|\langle x_k^{(N)}, x_j^{(N)}\rangle|^p\geq\Fc_{N-1}+2\sum_{j\neq k}|\langle x_k^{(N)}, x_j^{(N)}\rangle|^p.
%\sum_{i\neq j, i,j\in[1:N]}|\langle x_i^{(N)}, x_j^{(N)}\rangle|^p
\end{align}
Summing \eqref{equ:inc} over $k$, we obtain
$$N\Fc_N\geq N\Fc_{N-1}+2\Fc_N\Longrightarrow(N-2)\Fc_N\geq N\Fc_{N-1}\Longrightarrow\frac{\Fc_N}{N(N-1)}\geq\frac{\Fc_{N-1}}{(N-1)(N-2)}.$$
\end{proof}
 
It follows that  $\tau:=\lim_{N\to \infty}\frac{\Fc_{p,N}}{N^2}$ exists. In fact, in the minimal energy literature, $\tau$ is called the transfinite diameter due to Fekete.  Furthermore, $\tau$ is related to the continuous version of the frame potential, which is introduced in \cite{EhlOko2012}. More specifically, given a probabilistic measure $\mu$ on the sphere, the \emph{probabilistic $p$ frame potential} is defined as
\begin{equation}
\pfp_{p,d}(\mu):=\int_{\Sb^{d-1}}\int_{\Sb^{d-1}}|\langle x,y\rangle|^pd\mu(x)d\mu(y).
\end{equation}
Let $\mathcal{M}(\Sb^{d-1})$ be the collection of all probabilistic measures on the sphere. Simple compactness and continuity arguments show that 
\begin{equation}\label{equ:pfp}
\Pc_{p,d}:=\min_{\mu\in\mathcal{M}(\Sb^{d-1})}\pfp_{p,d}(\mu)
\end{equation}
exists. 

Given any $N$ point configuration $X$, its normalized counting measure is defined as $$\nu_{X}:=\frac{1}{N}\sum_{x\in X}\delta_x.$$
We have 
\begin{equation}\label{equ:relation}
\pfp_{p,d}(\nu_X)=\int\int |\langle x,y\rangle|^pd\nu_X(x)d\nu_X(y)=\frac{1}{N^2}\sum_{i=1}^N\sum_{j=1}^N|\langle x_i,x_j\rangle|^p=\frac{\fp_{p,N,d}(X)+N}{N^2}.
\end{equation}
Consequently, if $X$ is an optimal configuration, i.e.,  $\Fc_{p,N}=\fp_{p, N, d}(X)$,  then  by~\eqref{equ:relation}, it is plausible that $\tau= \Pc_{p,d}$. This is indeed the case, and it was proved in a more general setting by Farkas and Nagy \cite{FN08}. For the sake of completeness, we  reproduce their proof below.

\begin{lemma}\label{lem:1}
Given $d\geq 2$ and  $p\in(0,\infty)$, $\displaystyle \tau=\lim_{N\rightarrow\infty}\frac{\Fc_{p,N,d}}{N^2}\leq\Pc_{p,d}$.
\end{lemma}
\begin{proof}
Let $\mu^*$ be the optimal probabilistic measure, that is,
$$\int\int|\langle x,y\rangle|^pd\mu^*(x)d\mu^*(y)=\Pc_{p,d}=\pfp_{p,d}(\mu^*).$$ Consequently, 

\begin{align*}
\Fc_{p,N,d}&=\int\cdots\int \left[\min_X \fp(X)\right]d\mu^*(x_1)\cdots d\mu^*(x_N)\\
& \leq \int\cdots\int \fp(X)d\mu^*(x_1)\cdots d\mu^*(x_N)\\
&=\sum_{i\neq j}\int\cdots\int |\langle x_i,x_j\rangle|^p d\mu^*(x_1)\cdots d\mu^*(x_N)=\sum_{i\neq j}\Pc_{p,d}=N(N-1)\Pc_{p,d}.
\end{align*}
The result follows by dividing $N^2$ on both sides and taking the limit.
\end{proof}

We can now state and prove that $\Fc_{N}\sim N^2$ as $N\to \infty$. 

\begin{proposition}\label{prop:asymp}
Given $d\geq 2$ and  $p\in(0,\infty)$, we have $\displaystyle \tau= \lim_{N\rightarrow\infty}\frac{\Fc_{p,N,d}}{N^2}=\Pc_{p,d}$. Moreover, if $\{X_N\}_{N\geq d+1}$ is a sequence of $N$-point configurations such that $\lim_{N\rightarrow\infty}\frac{\fp_N(X_N)}{N^2}=\tau$, then every weak star cluster point $\nu^*$ of the normalized counting measure $\nu_{X_N}=\frac{1}{N}\sum_{x\in X_N}\delta_x$ solves \eqref{equ:pfp}, that is $\pfp_{p,d}(\nu^*)=\Pc_{p,d}$. In particular, this holds for any sequence of the optimal configurations of $\fp_N$.
\end{proposition}
\begin{proof}
By weak star convergence and \eqref{equ:relation}
\begin{align*}
\Pc_{p,d}\leq\pfp_{p,d}(\nu^*)=\lim_{N\rightarrow\infty}\pfp_{p,d}(\nu_{X_N})=\lim_{N\rightarrow\infty}\frac{\fp_{p,N,d}(X_N)+N}{N^2}=\tau.
\end{align*}
In view of Lemma \ref{lem:1}, we have $\tau=\Pc_{p,d}$ and $\nu^*$ is an optimal probabilistic measure.
\end{proof}

The exact value of $\tau$ can be found in many cases. We list two examples in the following corollary.

%When $p\in (0, 2)$ we can find the exact value of $\tau$.
\begin{corollary}

(a) When $d\geq 2$ and  $p\in(0,2]$, we have  $\displaystyle \lim_{N\rightarrow\infty}\frac{\Fc_{p,N,d}}{N^2}=\Pc_{p,d}=\frac{1}{d}$.

(b) When $d=2$ and $p$ is an even integer, we have $\displaystyle \lim_{N\rightarrow\infty}\frac{\Fc_{p,N,2}}{N^2}=\Pc_{p,2}=\frac{1\cdot3\cdot5\cdots (p-1)}{2\cdot4\cdot6\cdots p}$.
\end{corollary}
\begin{proof} 

(a) By  \cite[Theorem 3.5]{EhlOko2012} we know  that when $N=kd$, the frame potential is minimized by $k$ copies of orthonormal basis. So $\lim_{N\rightarrow\infty}\frac{\Fc_{p,N,d}}{N^2}=\lim_{k\rightarrow\infty}\frac{\Fc_{p,kd,d}}{(kd)^2}=\lim_{k\rightarrow\infty}\frac{(k-1)kd}{(kd)^2}=\frac{1}{d}$. Note that this recovers  \cite[Theorem 4.9]{EhlOko2012}, which states that $\Pc_{p,d}=\frac{1}{d}$.

(b) In dimension $d=2$, it is known that $2N$ equally spaced points on the unit circle are spherical $(2N-1)$-design (\cite[Section 4]{venkov2001reseaux}), so Proposition \ref{thm:spherical} implies that $X_N^{(h)}$ is an optimal configuration if $p\leq 2N-2$ is an even integer. In other words, with fixed even integer $p$, when $N$ is large enough, $\left(X_N^{(h)}\right)^{sym}$ is going to be a $(2N-1)$-design (hence $p$-design), so the equality in Proposition \ref{thm:spherical} holds and we get the desired result.
\end{proof}

%\begin{corollary}
%$p=2k$
%\end{corollary}

\section{Optimal configurations in dimension $2$}\label{sec3}

This section focuses on the case $d=2$, when the points are on the unit circle $\Sb^1\subset\R^2$. 

\subsection{A class of minimal energy problems}\label{subsec3.1}
 We recall that  when $N=2k$ is even and $0<p< 2$, the solution to~\eqref{equ:min} was given in \cite[Theorem 3.5]{EhlOko2012}, where it was established that the minimizers are $k$ copies of  any orthonormal basis of $\R^2$. The case $p=2$ was settled by Benedetto and Fickus \cite{BenFic2003}. 
 In order to address other values of $p$, we will consider a more general problem
\begin{equation}\label{equ:f}
\min_{X\subset C_r, |X|=N} \sum_{i\neq j}f(\|x_i-x_j\|^2),
\end{equation}
where $f: (0,4r^2]\rightarrow \R$ is a nonnegative and decreasing function, and $C_r$ is a $1-$dimensional circle with radius $r$. This circle $C_r$ does not need to be centered at 0 and could be in any dimension. It will become clear later why we require points on a general circle instead of the usual  $\Sb^1$.

The first result only requires $f$ to be convex, but it only works for up to 4 points.
\begin{theorem}\label{thm:4}
Given $r>0$, let $f:(0,4r^2] \rightarrow \R$ be a decreasing convex function. Any configuration $X_4^*$ of $4$ equally spaced points on $C_r$ is an optimal configuration of \eqref{equ:f}. If in addition, $f$ is strictly convex, then no other $4$-point configuration is optimal.
\end{theorem}

\begin{proof} 
Let $X_4=\{x_i\}_{i=1}^4$ be an arbitrary configuration with $x_i$ ordered counter clockwise. Let $\alpha_{ik}\in[0,2\pi)$ be the angle between $x_i$ and $x_{i+k}$ for any $k\in[1:3]$. The index of the vectors is cyclic as $x_i=x_{i-4}$. Then $\|x_i-x_{i+k}\|^2=2r^2-2r^2\cos\alpha_{ik}=4r^2\sin^2\frac{\alpha_{ik}}{2}.$ It is evident that $\sum_{i=1}^4\alpha_{ik}=2\pi k$.  Using the convexity of $f$,
\begin{align}\label{equ:convex}
\sum_{i\neq j}f\big(\|x_i-x_j\|^2\big)&=\sum_{k=1}^{3}\sum_{i=1}^4 f\big(\|x_i-x_{i+k}\|^2\big)=4\sum_{k=1}^{3}\frac{1}{4}\sum_{i=1}^4 f\big(\|x_i-x_{i+k}\|^2\big)\\ \notag 
&\geq 4\sum_{k=1}^{3}f\left(\frac{1}{4}\sum_{i=1}^4\|x_i-x_{i+k}\|^2\right)=4\sum_{k=1}^{3}f\left(\frac{4r^2}{4}\sum_{i=1}^4\sin^2\frac{\alpha_{ik}}{2}\right).
\end{align}

Next,   let $\beta_{ik}=\alpha_{ik}/2$. In order to minimize the right hand side of \eqref{equ:convex}, we solve
\begin{equation*}
 \max \sum_{i=1}^4 \sin^2 \beta_{ik} \qt{subject to}
	\beta_{ik}\geq0, \ \sum_{i=1}^4 \beta_{ik} = \pi k. 
\end{equation*}

When $k=1$, we let $\beta_i=\beta_{i1}$ for short. Using Lagrange multipliers, we have
  $0=\frac{\partial}{\partial\beta_{j}}[\sum_{i=1}^4\sin^2\beta_{i}+\lambda(\sum_{i=1}^4\beta_{i}-\pi)]=\sin2\beta_{j}+\lambda$, which implies that
 
$$ \sin2\beta_{i}=\sin2\beta_{j}\Longrightarrow 2\beta_{i}=2\beta_{j}, \text{ or } 2\beta_i+2\beta_j=\pi, $$ since $\sum_{i=1}^4\beta_i=\pi$. 
 
 If we are in the case that $\beta_1+\beta_2=\pi/2$ (or any pair $i\neq j$ with $\beta_i+\beta_j=\pi/2$), then
 $\sum_{i=1}^4 \sin^2 \beta_{i}=\sin^2(\beta_1)+\sin^2(\pi/2-\beta_1)+\sin^2(\beta_3)+\sin^2(\pi/2-\beta_3)=2$.
  If we are in the other case that $\beta_1=\beta_2=\beta_3=\beta_4$,  then $\sum_{i=1}^4 \sin^2 \beta_{i}=4\sin^2\frac{\pi}{4}=2.$
  So  for $k=1$,
 $$\sum_{i=1}^4\sin^2\beta_{i1}\leq 4\sin^2\frac{\pi}{4},$$
 and the equality holds when $\beta_{i1}+\beta_{j1}=\pi/2$ for some $i\neq j$.

When $k=2$, it is obvious that 
$$\sum_{i=1}^4\sin^2\beta_{i2}\leq 4=4\sin^2\frac{\pi}{2}$$ 
with equality at $\beta_{i2}=\pi/2$, for all $i\in[1:4]$. This implies that $\beta_{i1}+\beta_{i+1,1}=\pi/2$ for some $i$.
 
 When $k=3$, $\sum_{i=1}^4\sin^2\beta_{i3}=\sum_{i=1}^4\sin^2(\pi-\beta_{i1})=\sum_{i=1}^4\sin^2\beta_{i1}$ which reduces to the $k=1$ case.
 
 In summary, for any $k=1, 2,3$,
 
 $$\sum_{i=1}^4\sin^2\frac{\alpha_{ik}}{2}\leq 4\sin^2\frac{\pi k}{4},$$
and the equality holds simultaneously when $\alpha_{i1}+\alpha_{i+1,1}=\pi$, or equivalently $x_1+x_3=0, x_2+x_4=0$.

Following \eqref{equ:convex}, we have
\begin{align}\label{equ:f2}
\sum_{i\neq j}f\big(\|x_i-x_j\|^2\big)\geq 4\sum_{k=1}^{3}f\Big(\frac{4r^2}{4}\sum_{i=1}^4\sin^2\frac{\alpha_{ik}}{2}\Big)\geq 4\sum_{k=1}^{3}f\Big(4r^2\sin^2\frac{\pi k}{4}\Big)=8f(2r^2)+4f(4r^2).
\end{align}
It is easy to check that four equally spaced points on $C_r$ achieve this minimum. 

If $f$ is strictly convex, then the inequality of \eqref{equ:convex} becomes equality if $\|x_i-x_{i+k}\|=\|x_j-x_{j+k}\|$ for every $i\neq j$, which only holds for equally spaced points.
\end{proof}

\begin{remark}\label{rem:fail}
The proof of Theorem \ref{thm:4} breaks down for $N\geq5$ because $\sum_{i=1}^N \sin^2 \beta_{i1}$ is not  maximized at equally spaced points.
\end{remark}

Our  second result regarding \eqref{equ:f} is a variation of the main result of the work by Cohn and Kumar \cite[Theorem 1.2]{CohKum2007}. Let $m$ be a positive integer. An \emph{$m$-sharp configuration} $X\subset \Sd$ is  a spherical $(2m-1)$-design with $m$ inner products between its distinct points. 
It was proven in \cite{CohKum2007} that sharp configurations are the unique universal optimal configurations of the problem
\begin{equation}\label{equ:fd}
\min_{X\in S(N,d)} \sum_{i\neq j}f(\|x_i-x_j\|^2),
\end{equation}
for completely monotonic functions $f$.
A $C^{\infty}$ function $f: I\rightarrow\R$ is called \emph{$K$-completely monotonic} if $(-1)^kf^{(k)}(x)\geq 0$ for all $x\in I$ and all $k\leq K$, and \emph{strictly $K$-completely monotonic} if strict inequality always holds in the interior of $I$. 
%Using these notions, it was proved in \cite{CohKum2007} that  sharp configurations are  universal minimizers of the above potential. In addition, 
The notion $\infty$-completely monotonic is simply called \emph{completely monotonic} as traditionally defined, which means $(-1)^kf^{(k)}(x)\geq 0$ for all $x\in I$ and all $k\geq0$.
A list of known sharp configurations was given in \cite{CohKum2007}.
For example, $N$ equally spaced points on $\Sb^1$ is an  $\lfloor{N/2}\rfloor$-sharp configuration.

Another notion that we will need is that of  absolutely monotonic functions. A $C^{\infty}$ function $f: I\rightarrow\R$ is called \emph{$K$-absolutely monotonic} if $f^{(k)}(x)\geq 0$ for all $x\in I$ and all $k\leq K$. Similarly, $\infty$-absolutely monotonic means the inequality is true for all nonnegative integers $k$, and will be simply referred to as \emph{absolutely monotonic}. It is straightforward that $f(t)$ being completely monotonic is equivalent to $f(-t)$ being absolutely monotonic. %Furthermore, a function $f$ has one of these properties \emph{up to $M$} if the appropriate inequalities hold for $k\leq M$.

As remarked by \cite{CohKum2007}, the complete monotonicity on $f$ can be weakened slightly. To ensure a good flow of the paper,  the proof of the next result  which is a variation of \cite[Theorem 1.2]{CohKum2007} will be given in the appendix.

\begin{theorem}\label{thm:CK}
Fix a positive integer $m$ and let $f:(0,4]\rightarrow\R$ be a function such that $(-1)^kf^{(k)}(t)\geq0$ for all $t\in(0,4], k\leq 2m$. Then an $m$-sharp configuration is an optimal configuration of \eqref{equ:fd}. Furthermore, if $(-1)^kf^{(k)}(t)>0$ for all $t\in(0,4), k\leq 2m$, then the $m$-sharp configuration is the unique optimal configuration of \eqref{equ:fd}.
\end{theorem}

%We will define this weakened monotonicity formally for convenience.
%Given a natural number $K$, a function $f: I \rightarrow \R$ is called \emph{$K$-completely monotonic} or completely monotonic up to $K$ if $(-1)^kf^{(k)}(r)\geq0$ for all $r\in I, k\leq K$. Furthermore, it is called \emph{strictly $K$-completely monotone} if strict inequality always holds in the interior of $I$.
%%If $g(t)=f(-t)$, then complete monotonicity becomes absolute monotonicity. Therefore it is also convenient to define absolute monotonicity up to an order. 
%A function $f: I \rightarrow \R$ is called \emph{$K$-absolutely monotonic} or absolutely monotonic up to $K$ if $f^{(k)}(r)\geq0$ for all $r\in I, k\leq K$. Furthermore, it is called \emph{strictly $K$-absolutely monotone} if strict inequality always holds in the interior of $I$.

A direct consequence of Theorem \ref{thm:CK} for dimension $d=2$ is that equally spaced points are optimal configurations if the energy kernel function $f$ is completely monotonic up to certain order. But notice that $\sum_{i\neq j}f(\|x_i-x_j\|^2)$ only depends on the relative distances between $x_i$'s so the result should be true for any circle $C_r$ (whose radius is $r$) if we rescale $f$ properly.

\begin{corollary}\label{cor:CK}
For $N\geq 4$, let $m=\lfloor{N/2}\rfloor$. For $r>0$, suppose that  $f: (0,4r^2]\rightarrow \R$ is completely monotonic up to $2m$.  Then $N$ equally spaced points on $C_r$ is an optimal configuration of \eqref{equ:f}. Moreover, if $f$ is strictly completely monotonic up to $2m$, then the equally spaced points is the unique optimal configuration of \eqref{equ:f}.
\end{corollary}

\subsection{A lifting trick}\label{subsec3.2}
How do Theorem \ref{thm:4} and Corollary \ref{cor:CK} help to solve the frame potential problem? On the unit circle, we have $|\langle x_i,x_j\rangle|^p=\left|\frac{2-\|x_i-x_j\|^2}{2}\right|^p=h(\|x_i-x_j\|^2),$ where $h(t)=\left|\frac{2-t}{2}\right|^p$. Unfortunately neither result can be applied because the function $h(t)$ is not differentiable at $t=2$ unless $p$ is an even integer; worse, it is not even decreasing on [0,4]. This should not come as a surprise since the frame potential does not distinguish between antipodal points. Consequently, rather than analyzing the frame potential in terms of the distance between vectors, we should consider it in terms of the distance between lines, as was done in \cite{CP16}.

Define $P:\Sd\rightarrow M(d,d)$ as 
$P(x)=xx^*,$ where $M(d,d)$ is the space of $d\times d$ symmetric matrices endowed with the Frobenius norm.
$P(\Sd)$ identifies antipodal points, and is the projective space embedded in $M(d,d)$. 
We write $P(x)$ as $P_x$ and  list some of the properties.
\begin{equation}\label{equ:pr}
\left\{ \begin{array} {l@{\ =\ } l}
\langle P_x, P_y\rangle &|\langle x,y\rangle|^2\\
\|P_x-P_y\|^2&2-2|\langle x,y\rangle|^2
.\end{array}\right.
\end{equation}
When $d=2$, we can explicitly write the embedding as
 $P:\Sb^1\rightarrow M(2,2) (= \R^3),$
$$P(x)=P_x=xx^*=\left[\begin{array}{cc}x_1^2 &x_1x_2 \\ x_1x_2 & x_2^2\end{array}\right]\longleftrightarrow(x_1^2,\sqrt{2}x_1x_2,x_2^2).$$
It is not hard to see that $P(\Sb^1)$ is a circle in $\R^3$ centered at $(\frac{1}{2},0,\frac{1}{2})$ with radius $r=\frac{1}{\sqrt{2}}$,  and this is where we can apply Theorem \ref{thm:4} or Corollary \ref{cor:CK}. One can verify that equally spaced points on the circle $P(\Sb^1)$ are precisely $X_N^{(h)}$, equally spaced points on the half circle, so we have the following theorem.

\begin{theorem}\label{thm:both}
Let $g:[0,1) \rightarrow \R$ and consider 
\begin{equation}\label{equ:gprojective}
\min_{X\in S(N,2)} \sum_{i\neq j}g(|\langle x_i,x_j\rangle|^2),
\end{equation}
Then the following statements hold.
\begin{enumerate}
\item[(a)] If $g$ is convex and increasing, then $X_4^{(h)}$ is an optimal configuration of \eqref{equ:gprojective} when $N=4$. Moreover if $g$ is strictly convex, then $X_4^{(h)}$ is the unique optimal configuration.
\item[(b)] If $g$ is absolutely monotone up to $2\lfloor{N/2}\rfloor$, then $X_N^{(h)}$ is an optimal configuration of \eqref{equ:gprojective}. Moreover if $g$ is strictly absolutely monotone up to $2\lfloor{N/2}\rfloor$, then $X_N^{(h)}$ is the unique optimal configuration.
\end{enumerate}
\end{theorem}

\begin{proof}
As defined, $P_{x_i}=x_ix_i^*$. Denote $P_{x_i}$ by $P_i$ for simplicity. By \eqref{equ:pr},
$$g(|\langle x_i,x_j\rangle|^2)=g(1-\|P_i-P_j\|^2/2)=:f(\|P_i-P_j\|^2),$$
where $f(t)=g(1-t/2)$ is defined on $(0,2]$. As discussed earlier, view the points $P_i$  on a circle in $\R^3$ with radius $1/\sqrt{2}$, so solving \eqref{equ:gprojective} is equivalent to solving \eqref{equ:f} with $r=1/\sqrt{2}$.

If $g$ is convex and increasing, then $f$ is convex and decreasing. Applying Theorem \ref{thm:4} gives equally spaced $P_i$, which is equally spaced points on the half circle. This is part (a).

If $g$ is absolutely monotone up to $2\lfloor{N/2}\rfloor$, then $f$ is completely monotone up to $2\lfloor{N/2}\rfloor$. Applying Corollary \ref{cor:CK} gives part (b).

\end{proof}

\begin{remark}
Observe that in Theorem \ref{thm:both}, the assumption of (b) is much stronger than (a). If $g$ is twice differentiable, then $g$ being convex and decreasing is equivalent to $g$ being absolutely monotone up to 2. Furthermore, Theorem \ref{thm:both} is a very general result that goes beyond frame potentials. Indeed, it cover the cases where the energy can be expressed as a function of squares of the inner products. We expect to pursue this line of investigations elsewhere, with the goal of analyzing other energy kernels  suitable for finding certain well conditioned frames.
\end{remark}

Finally we are ready to state the promised frame potential result as a special case of Theorem \ref{thm:both}.

\begin{theorem}\label{thm:p}
Let  $X^{(h)}_N$ be the equally spaced points on  half of the circle $\Sb^1$ as in \eqref{equ:half}. The  following statements hold.
\begin{enumerate}
\item[(a)] If $N=4$ and $p>2$, then $X^{(h)}_4$ is the unique optimal configuration of \eqref{equ:min}.
\item[(b)] If $N\geq5$ and $p>\left\{\begin{array}{ll}
2N-2, &\text{$N$ is even}\\
2N-4, &\text{$N$ is odd}
\end{array}\right.$, then $X^{(h)}_N$ is the unique optimal configuration of \eqref{equ:min}.
\item[(c)] If $N\geq 5$,  and $2<p\leq\left\{\begin{array}{ll}
2N-2, &\text{$N$ is even}\\
2N-4, &\text{$N$ is odd}
\end{array}\right.$ is an even integer, then $X^{(h)}_N$ is an optimal configuration of \eqref{equ:min}, but it is unclear whether there are other optimal configurations.
\end{enumerate}
\end{theorem}

\begin{proof}

The $p$-frame potential kernel $|\langle x_i,x_j\rangle|^p=g_p(|\langle x_i,x_j\rangle|^2)$ with $g_p(t)=t^{p/2}$.
 The function $g_p$ is strictly convex and increasing on [0,1) if $p>2$.

\begin{enumerate}
\item[(a)] This part is  due to Theorem \ref{thm:both}(a).
\item[(b)] We notice that $g_p$ is strictly absolutely monotone up to $\lceil p/2\rceil$, where $\lceil c\rceil$ is the smallest integer that is no less than $c$. In order to apply Theorem \ref{thm:both}(b), we require $\lceil p/2\rceil\geq 2\lfloor{N/2}\rfloor$, which is equivalent to $p>2N-2$ if $N$ is even and $p>2N-4$ if $N$ is odd.
\item[(c)]  Finally, this part is true because $g_p$ is absolutely monotone when $p$ is an even integer.
\end{enumerate}
 \end{proof}

\begin{remark}
By Proposition \ref{prop:asymp}, we can let $p$ go to infinity in Theorem \ref{thm:p} and get that $X_N^{(h)}$ is the Grassmannian frame, as was shown in \cite{BK06}.

As seen, the 1-dimensional projective space is isomorphic to a circle. It is well known that higher a dimensional projective space is not a higher dimensional sphere. This is  why the main result Theorem \ref{thm:both} is limited to $d=2$.
\end{remark}

At this point, we summarize the $p$-frame potential results in $\Sb^1$ as the following remark.

\begin{remark}\label{rem:N}Let $d=2$.\textcolor{white}{l}
\begin{enumerate}
\item[(a)] When $N=4$ we have completed the characterization of $\Fc_{p,4,2}$. 
\item[(b)] When $N\geq6$ is even, then \cite[Theorem 3.5]{EhlOko2012} and parts (b) and (c) of Theorem~\ref{thm:p} give the value of $\Fc_{p,N, 2}$ when $p\in (0, 2] \cup\{4, 6, \cdots, 2N-2\}\cup (2N-2, \infty)$. We further know that the minimizer is unique for $p\in (0, 2) \cup (2N-2, \infty)$. It is still open for $p\in (2, 2N-2]$ though we expect $X^{(h)}_N$ to  be a minimizer. The numerical result is displayed in Figure \ref{fig:56} for $N=6$.

\item[(c)]  When $N\geq5$ is odd, we know $\Fc_{p, N, 2}$ for $p\in \{2,4,\cdots, 2N-4\}\cup (2N-4, \infty)$. We suspect that for $p\in (2, 2N-4]$, $X^{(h)}_N$ will still be the minimizer. The case $p\in (0, 2)$ seems rather intriguing as demonstrated in Figure \ref{fig:56} for $N=5$.  
\end{enumerate}
\end{remark}

Figure \ref{fig:56} displays the numerical experiment for $d=2$ and $N=5,6$. According to the numerical experiment, $\Fc_{p,6,2}$ is achieved by $X_6^{(h)}$ for $p\in(2,10]$. The $N=5$ case is more complex. It appears that  for $p$ from 0 to about 1.78, the optimal configuration is two copies of ONB plus a repeated vector; for $p\in(1.78,2)$, the optimal configuration has the structure $\{x,x,y,y,z\}$ whose angles vary as $p$ changes; for $p\in(2,6)$, the optimal configuration is $X_5^{(h)}$
\begin{figure}[hbt]
\centering
\includegraphics[width=0.95\textwidth]{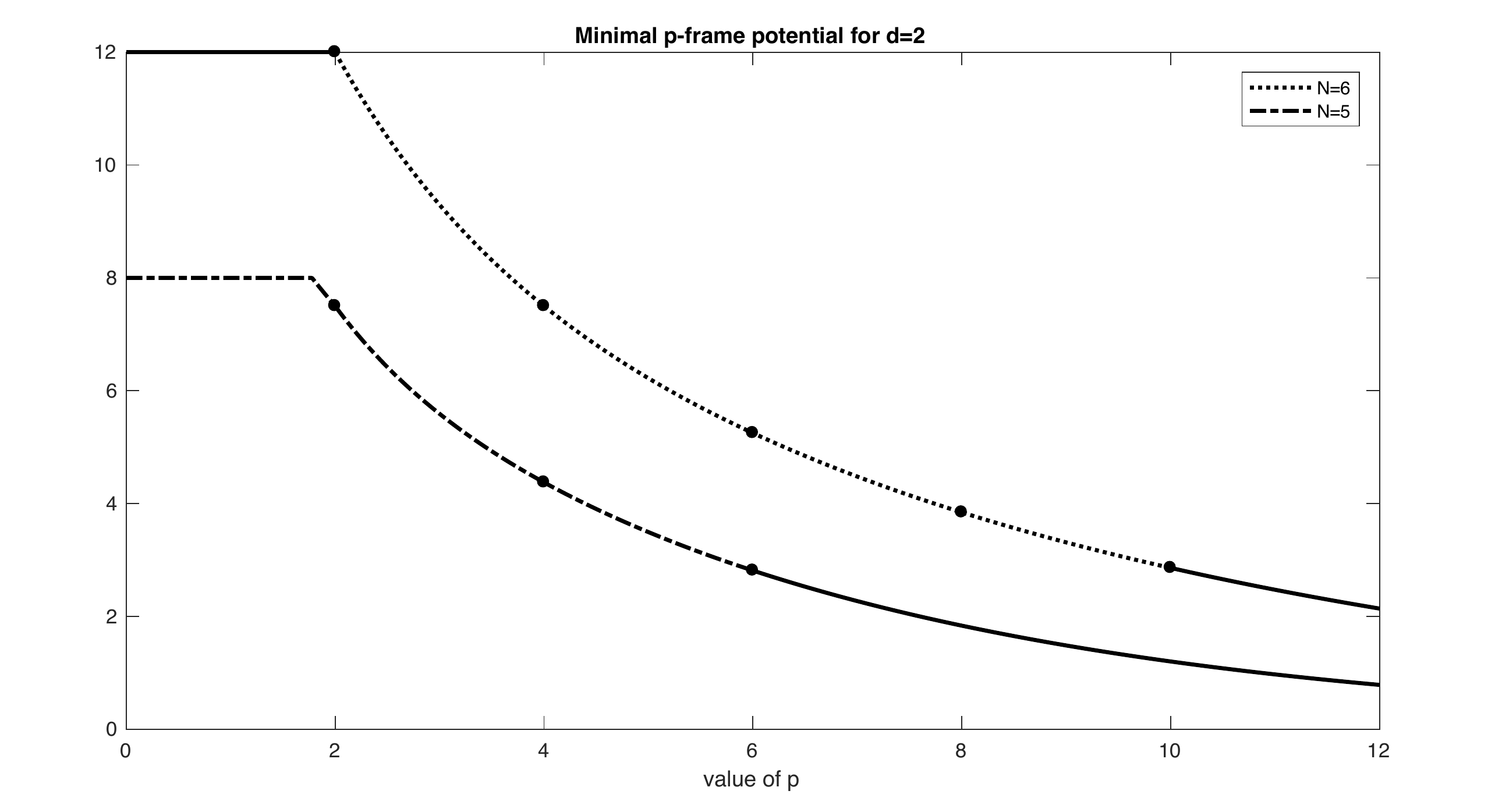}
\caption{The top curve represents $\Fc_{p,6,2}$ while the lower one represents $\Fc_{p,5,2}$. The solid portion  indicates proven cases as commented in Remark \ref{rem:N}. }\label{fig:56}
\end{figure}

\section{Special case of $N=d+1$ points in dimension $d$. }\label{sec:d+1}

In this last section, we report on some numerical experiments and the resulting conjectures when minimizing the $p$-frame potential with $N=d+1$ vectors in $\R^d$,  and $p\in (0, \infty)$. Observe that the case $p = 2$ is a special case of the work by Benedetto and Fickus~\cite{BenFic2003}. Additionally, the case $p > 2$ is handled by Ehler and Okoudjou~\cite[Proposition~3.1]{EhlOko2012}, for which the simplex is the optimal configuration. To be specific, the simplex is an ETF of $d+1$ vectors for $\R^d$. Therefore, the focus in this section are values $p < 2$. The following definition will be used through the rest of this section.

\begin{definition} \label{def:onbp}\textcolor{white}{l}
%Let $X=\{x_k\}_{k=1}^{d+1}\subset \Sd$. 
\begin{enumerate}
\item[(a)] $X$ is an ONB+ if $X$ is formed by  an orthonormal basis of $\R^d$ with one of the vectors repeated.
\item[(b)] Given $n\geq2$,  the simplex of $\R^n$ is denoted by ETF$_n$. An explicit construction of ETF$_n$ is to project $e_1, e_2, \cdots, e_n, e_{n+1}$, the canonical basis of $\R^{n+1}$, onto the orthogonal complement of $\sum_{i=1}^{n+1}e_i$.
\end{enumerate}

\end{definition}

%\textcolor{blue}{define $ETF_k$, and define $L_1^1$}

\subsection{Lifted ETFs}
From numerical tests, we have noticed that minimizers for $\Fc_{p, d+1,d}$ take forms similar to ETFs. In particular, they take the form of ETFs that have been lifted to higher dimensions.
\begin{definition}
  \label{defn:L-frames}
  For $1 \le k \le d$,  the frame
  \begin{align*}
    \text{L}_k^d = \begin{bmatrix}
      \text{ETF}_k & 0 \\
      0 & I_{d-k}
    \end{bmatrix}
    = \begin{bmatrix}
      \text{ETF}_k & 0 & \cdots & 0 \\
      0 & 1 & \cdots & 0 \\
      \vdots & \vdots & \ddots & \vdots \\
      0 & 0 & \cdots & 1
    \end{bmatrix}\in S(d+1,d)
  \end{align*}
  is called a \emph{lifted ETF}.
\end{definition}

\begin{remark}
  \label{remark:after-defn-of-L-frames}\textcolor{white}{l}
  \begin{enumerate}
    \item[(a)] In Definition~\ref{defn:L-frames}, the entry ETF$_k$ is the synthesis operator for the ETF$_k$ configuration, and $I_{d-k}$ is the $(d-k) \times (d-k)$ identity matrix. These frames are lifted in the sense that unit vectors for the remaining dimensions ($e_{k+1},  e_{k+2}, \dotsc$, and $e_d$) have been added such that the ETF$_k$ frame is moved from $\R^k$ to $\R^d$. We refer to \cite{trem2009, tremain_II} for more on constructions of these classes of ETFs.

    \item[(b)] The L$_k^d$ frames are not tight, except for the case $k = d$, and we have L$_d^d$ is ETF$_d$.

    \item[(c)] In addition to considering ETF$_d$ as an L$_d^d$ configuration,  ONB+ is the L$_1^d$ frame.
  \end{enumerate}
\end{remark}

\begin{example}
  \label{example:MBz}
  The ETF in $\R^2$ can be lifted to $\R^3$ as
  \begin{figure}[h]
    \begin{tikzpicture}
      % Draw MBz
      \pgfmathsetmacro{\scale}{1.5}
      \draw[->] (0.000000, 0.000000) -- (-0.989949*\scale, -0.329983*\scale);
      \draw[->] (0.000000, 0.000000) -- (0.989949*\scale, -0.329983*\scale);
      \draw[->] (0.000000, 0.000000) -- (0.000000*\scale, 1.319933*\scale);
      \draw[->, ultra thick] (0.000000*\scale, 0.000000*\scale) -- (-0.707107*\scale, -0.235702*\scale);
      \draw[->, ultra thick] (0.000000*\scale, 0.000000*\scale) -- (0.965926*\scale, -0.086273*\scale);
      \draw[->, ultra thick] (0.000000*\scale, 0.000000*\scale) -- (-0.258819*\scale, 0.321975*\scale);
      \draw[->, ultra thick] (0.000000*\scale, 0.000000*\scale) -- (0.000000*\scale, 0.942809*\scale);
    \end{tikzpicture}
    \caption{The L$_2^3$ frame}
  \end{figure}
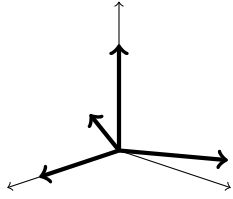
  \begin{align*}
    \text{L}_2^3 =
    \begin{bmatrix}
      1 & -1/2 & -1/2 & 0 \\
      0 & \sqrt{3}/2 & -\sqrt{3}/2 & 0 \\
      0 & 0 & 0 & 1
    \end{bmatrix}.
  \end{align*}
  We see that this frame is neither tight nor equiangular by computing the frame operator and Grammian,
  \begin{align*}
    S = \begin{bmatrix}
      3/2 & 0 & 0 \\
      0 & 3/2 & 0 \\
      0 & 0 & 1
    \end{bmatrix}
    \qquad
    G = \begin{bmatrix}
      1 & -1/2 & -1/2 & 0 \\
      -1/2 & 1 & -1/2 & 0 \\
      -1/2 & -1/2 & 1 & 0 \\
      0 & 0 & 0 & 1
    \end{bmatrix}.
  \end{align*}
\end{example}

More generally, the Grammian of the L$_k^d$ frame is
  \begin{align}\label{equ:grammian}
    \begin{bmatrix}
      1 & -1/k & -1/k & \cdots & -1/k & 0 & 0 & \cdots & 0 \\
      -1/k & 1 & -1/k & \cdots & -1/k & 0 & 0 & \cdots & 0 \\
      -1/k & -1/k & 1 & \cdots & -1/k & 0 & 0 & \cdots & 0 \\
      \vdots & \vdots & \vdots & \ddots & \vdots & \vdots & \vdots & \ddots & \vdots \\
      -1/k & -1/k & -1/k & \cdots & 1 & 0 & 0 & \cdots & 0 \\
      0 & 0 & 0 & \cdots & 0 & 1 & 0 & \cdots & 0 \\
      0 & 0 & 0 & \cdots & 0 & 0 & 1 & \cdots & 0 \\
      \vdots & \vdots & \vdots & \ddots & \vdots & \vdots & \vdots & \ddots & \vdots \\
      0 & 0 & 0 & \cdots & 0 & 0 & 0 & \cdots & 1 \\
    \end{bmatrix}
  \end{align}
  indicating that each L$_k^d$ frame is a two-distance set (see, \cite{Delsarte:1977aa, barg_etal}) with inner products~$-1 / k$ and~$0$; note, however, that L$_d^d$, or the ETF$_d$ configuration, will have only one inner product,~$-1 / d$.

%\begin{remark}
%  \label{remark:inner-products-of-ETF-in-L}
%  The inner products corresponding to ETF$_k$ in the Grammian are all negative because we assume a particular ETF$_k$ is used to build the L$_k^d$ frame (see Remark~\ref{remark:after-defn-of-L-frames}).
%\end{remark}

\subsection{Lifted ETFs as the Conjectured Minimizers}
Numerical computations suggest that the L$_k^d$ frames are minimizers of $\fp_{p, d+1, d}$. 
\begin{conjecture}\label{conj}
  \label{conjecture:L-frames-are-minimizers} Suppose $d\geq 2$ and for every natural number $1\leq k \leq d-1$, let
  \begin{align*}
    p_k = \frac{\log(k+2) - \log k}{\log(k+1) - \log k}.
  \end{align*}
  We also define $p_0=0$. The following configurations minimize the $p$-frame potential $FP_{p, d+1, d}$:
  \begin{itemize}
   % \item when $p \in [0, p_1]$, the ONB+, or L$_1^d$, configuration,
    \item when $p \in (p_{k-1}, p_k]$, the L$_k^d$ configuration, $k=1, 2, \cdots, d-1$;
    \item when $p \in (p_{d-1}, \infty]$, the ETF$_d$, or L$_d^d$ configuration.
  \end{itemize}
\end{conjecture}

Figure \ref{fig:d+1} visualizes this conjecture. Note that part of this conjecture is already proven.
The case $d=2$ is completely established in~\cite{EhlOko2012}. 
For $d\geq3$, the statement that ETF$_d$ is the minimizer  follows from~\cite{BenFic2003} when $p = 2$,  from~\cite[Proposition~3.1]{EhlOko2012} when $p>2$, and from \cite{{HeaStr2003}} for $p=\infty$. A.~Glazyrin~\cite{Glaz19}  recently established that the ONB+, or L$_1^d$ is the optimal configuration for $p\in(1,2(\frac{\ln3}{\ln2}-1)]$, leading to the fact that $\Fc_{p, d+1, d}=2$ for all $p$ in this range and all $d\geq 2$. The number $2(\frac{\ln3}{\ln2}-1)$ is approximately 1.17 and is less than $p_1$.  Conjecture~\ref{conj} was tested numerically for $d = 3, 4, 5, 6, 7$; details are given in Section~\ref{section:numerical-info}.

\begin{figure}[hbt]
\centering
\begin{tikzpicture}[scale=1.6]
\draw[->, thick] (-2.5,0)--(-2.5,2.4) node[left]{$\Fc_{p, d+1,d}$};
 \draw[gray]  (-2.5,2)--(-2.57,2) node[left, black] {2};
\draw[->, thick](-2.5,0)--(6.5,0) node[below]{$p$};
\draw (-2.5,0) node[below] {0};
%\draw[gray] (-3.8,0)--(-3.8,-0.07) node[below, black]{1};
\draw[gray] (0.7+0.5,0)--(0.7+0.5,-0.07) node[below, black]{1.5};
\draw[gray] (1.7+0.5,0)--(1.7+0.5,-0.07) node[below, black]{1.6};
\draw[gray] (2.7+0.5,0.)--(2.7+0.5,-0.07) node[below, black]{1.7};
\draw[gray] (3.7+0.5,0)--(3.7+0.5,-0.07) node[below, black]{1.8};
\draw[gray] (4.7+0.5,0)--(4.7+0.5,-0.07) node[below, black]{1.9};
\draw[gray] (5.7+0.5,0)--(5.7+0.5,-0.07) node[below, black]{2};
\draw[red, ultra thick] (-1.3,2)--({1.2+10*(ln(3)/ln(2)-1.5)},2);
\draw[red, ultra thick] (-2.5,2)--(-2,2);
\draw[black, ultra thick] (-2,2)--({1.2+10*(1.25-1.5)},2);
\draw[dashed](-1.3,2)--(-1.3,0) node[below]{$2(\frac{\ln3}{\ln2}-1)$};
\draw[red] (-0,2) node[above]{$\rm{L}_1^d$};
\draw[dashed]({1.2+10*(ln(3)/ln(2)-1.5)},2)--({1.2+10*(ln(3)/ln(2)-1.5)},0) node[above]{$p_1$};
\draw [orange,ultra thick]  plot [domain={1.2+10*(ln(3)/ln(2)-1.5)}:{2.2+10*(ln(2)/ln(1.5)-1.6)}] (\x,{6*0.5^(1.5+0.1*(\x-1.2))});
\draw[orange] (2.7,1.9) node[above]{$\rm{L}_2^d$};
\draw[dashed]({2.2+10*(ln(2)/ln(1.5)-1.6)},1.82)--({2.2+10*(ln(2)/ln(1.5)-1.6)},0) node[above ]{$p_2$};
\draw [darkgreen,ultra thick]  plot [domain={2.2+10*(ln(2)/ln(1.5)-1.6)}:{1.2+10*(ln(5/3)/ln(4/3)-1.5)}] (\x,{12*0.333^(1.5+0.1*(\x-1.2))});
\draw[darkgreen] (3.65,1.75) node[above]{$\rm{L}_3^d$};
\draw[dashed]({1.2+10*(ln(5/3)/ln(4/3)-1.5)},1.7)--({1.2+10*(ln(5/3)/ln(4/3)-1.5)},0) node[above ]{$p_3$};
\draw [blue,ultra thick]  plot [domain={1.2+10*(ln(5/3)/ln(4/3)-1.5)}:{1.2+10*(ln(6/4)/ln(5/4)-1.5)}] (\x,{12*0.333^(1.5+0.1*(\x-1.2))});
\draw[blue] (4.17,1.65) node[above]{$\rm{L}_4^d$};
\draw[dashed]({1.2+10*(ln(6/4)/ln(5/4)-1.5)},1.63)--({1.2+10*(ln(6/4)/ln(5/4)-1.5)},0) node[above]{$p_4$};
\draw[dashed, ultra thick, yellow] ({1.2+10*(ln(6/4)/ln(5/4)-1.5)},1.63)--({2.2+10*(ln(12/10)/ln(11/10)-1.6)},1.45);
\draw [purple,ultra thick]  plot [domain={2.2+10*(ln(12/10)/ln(11/10)-1.6)}:{6.2}] (\x,{12*0.333^(1.5+0.1*(\x-1.2))});
\draw [black,ultra thick]  plot [domain={6.2)}:{6.6}] (\x,{12*0.333^(1.5+0.1*(\x-1.2))});
\draw[dashed]({2.2+10*(ln(12/10)/ln(11/10)-1.6)},1.45)--({2.2+10*(ln(12/10)/ln(11/10)-1.6)},0) node[above]{$p_{d-1}$};
\draw[purple] (6,1.37) node[above]{$\rm{L}_d^d/\rm{ETF}_d$};
\end{tikzpicture}
\caption{Conjectured optimal configurations for $\mathcal{F}_{p,d+1,d}$ as $p$ increases from 0 to $\infty$. The black lines are cases already proven.}\label{fig:d+1}
\end{figure}
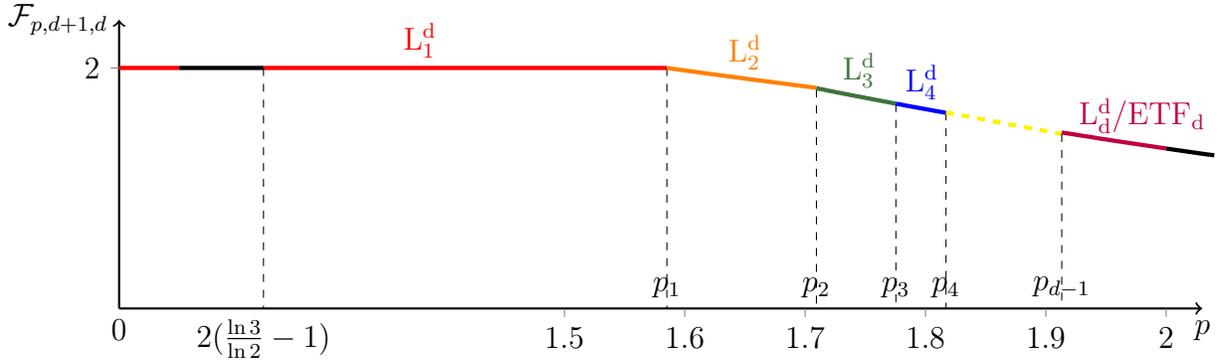

%\begin{figure}[h]
%  \centering
%  %Title:
%  $\Fc_{p, d+1, d}$ for $d = 2$ (Proven in~\cite{EhlOko2012}) and $d = 3, 4, 5$ (Conjectured)\\[0.2cm]
%   \includegraphics{figure-LkN-conjecture.pdf}
%  \caption{The conjectured minima of the $p$-frame potential, $\Fc_{p, d+1, d}$, with the minimizing configurations for $d = 3, 4, 5$ compared to the known (see~\cite{EhlOko2012}) minima, $\Fc_{p, 3, 2}$, for $d = 2$.}
%\end{figure}

The values $p_k$ may be found by using the $p$-frame potentials of the L$_k^d$ frames. By \eqref{equ:grammian},
%By~\cite[Theorem~2.3]{HeaStr2003},
%\begin{align*}
%  |\langle \phi_i, \phi_j \rangle| = \sqrt{\frac{k + 1 - k}{k (k+1 - 1)}} = \frac{1}{k}
%\end{align*}
%for any distinct $\phi_i$ and $\phi_j$ in ETF$_k$; we then have
\begin{align*}
  \fp_{p, d+1, d}(\text{L}_k^d) = \left((k+1)^2 - (k+1)\right)\left(\frac{1}{k}\right)^p =  (k+1) k \left(\frac{1}{k}\right)^p.
\end{align*}
We find $p_k$ so that the $p$-frame potentials of L$_k^d$ and L$_{k+1}^d$ are equal at the value $p_k$, so
$$(k + 1) k \left( \frac{1}{k} \right)^{p_k} = (k + 2) (k + 1) \left( \frac{1}{k+1} \right)^{p_k} $$  leads to 
 % &\iff \frac{k}{k+2} = \left( \frac{k}{k+1} \right)^{p_k} \\
  $p_k = \frac{\log(k+2) - \log k}{\log(k+1) - \log k}.$
%\end{align*}
\begin{remark}\textcolor{white}{l}
  \begin{enumerate}
    \item[(a)] The value $p_k$, where the $p$-frame potential of the L$_{k+1}^d$ frame drops below the $p$-frame potential of the L$_k^d$ frame, does not depend on $d$, the overall dimension.
    \item[(b)] Following Conjecture~\ref{conjecture:L-frames-are-minimizers}, we will call the values $p_k$ are the \emph{switching points} as these are the values of $p$ where the minimizing configuration seems to change. The final switching point is approaching to $2$:
   % \item Each time $d$ increases there seems to be a new switching point, $p_k$, and a new lifted minimizer that is not tight (see Remark~\ref{remark:after-defn-of-L-frames}). By the work of Benedetto and Fickus~\cite{BenFic2003}, however, the minimizer must be tight when $p = 2$. It is promising that the final switching point will never reach~$2$, although it will approach~$2$:
    \begin{align*}
      \lim_{d \to \infty} p_{d-1}
      &= \lim_{d \to \infty} \frac{\log\left( \frac{d+1}{d-1} \right)}{\log\left( \frac{d}{d-1} \right)} =2.
      %\\
   %   &= \lim_{d \to \infty} \frac{\frac{d-1}{d+1} \cdot \frac{(d-1) - (d+1)}{(d-1)^2}}{\frac{d-1}{d} \cdot \frac{(d-1) - d}{(d-1)^2}} \\
%      &= \lim_{d \to \infty} \frac{\frac{-2}{d+1}}{\frac{-1}{d}} \\
%      &= \lim_{d \to \infty} \frac{2 d}{d+1} \\
%      &= 2.
  \end{align*}
  \end{enumerate}
\end{remark}

\subsection{Description of the Numerical Computations}
\label{section:numerical-info}
Here we describe the program we wrote to test Conjecture~\ref{conjecture:L-frames-are-minimizers} numerically. The program was written using Sage \cite{sagemath} and tested this conjecture for $d+1$ vectors in $\R^d$ with $d = 3, 4, 5, 6, 7$, and with $k = 1, 2, \dotsc, d$. Given a frame L$_k^d$ to test, the program proceeds as follows.
\begin{enumerate}
  \item Determine the interval $[p_{k-1}, p_k]$ to test. The following were special cases:
  \begin{itemize}
    \item In the case of ONB+, the program tested the interval $[0.1, p_1],$ since minimization near $p = 0$ is difficult for the algorithm; furthermore, knowing that ONB+ is the minimizer at $p = 0.1$ is enough to know that it is the minimizer for all $p \in [0, 0.1]$.
    \item In the case of ETF$_d$, the program tested $[p_{d-1}, 2]$.
  \end{itemize}
  In the remainder, we refer to the lower and upper bounds of the tested interval as $p_{\text{min}}$ and $p_{\text{max}}$.
  
  \item Create a variable to count the number of frames with a lower $p$-frame potential than the L$_k^d$ frame, and set this count to zero.
  
  \item \label{numerical-procedure:frame-comparison-step} Generate a random frame, and use a minimization algorithm to decrease the $p_{\text{min}}$ frame potential. Compare the minimized frame potential to that of the L$_k^d$ frame; if it is strictly lower, increase the counter. Each time this step was repeated, a new frame was randomly generated.
  
  \item \label{numerical-procedure:repeat-step} Repeat Step~\ref{numerical-procedure:frame-comparison-step} four more times.
  
  \item Perform Steps~\ref{numerical-procedure:frame-comparison-step} and~\ref{numerical-procedure:repeat-step} with $p_{\text{max}}$ instead of $p_{\text{min}}$.
  
  \item \label{numerical-procedure:random-step} Perform Step~\ref{numerical-procedure:frame-comparison-step} with a random value $p^* \in [p_{\text{min}}, p_{\text{max}}]$ instead of $p_{\text{min}}$.
  
  \item Repeat Step~\ref{numerical-procedure:random-step} thirty-nine more times.
  
  \item Display the number of frames with $p$-frame potentials lower than that of the L$_k^d$ frame.
\end{enumerate}
In summary, the procedure makes fifty comparisons to the $p$-frame potential of the L$_k^d$ frame, ten of which are on the endpoints of the interval $[p_{\text{min}}, p_{\text{max}}]$. The result of this program for each pair $k$, $d$ tested was that only one frame was found with a strictly lower frame potential, but the difference was within the realm of numerical error (\verb"<1e-15").
%\footnote{See test \#27 in the file \texttt{L\_4\^{}4 Conjecture-500.txt}.}

\section{Appendix: Proof of Theorem \ref{thm:CK}}
We now give a proof of Theorem~\ref{thm:CK}. Let $f$ be a smooth function. 
Given a polynomial $g$ with $\deg(g) \geq1$, let $H(f,g)$ denote the Hermite interpolating polynomial of degree less than $\deg(g)$ that agrees with $f$ at each root of $g$ to the order of that root. The following fact is proven in the proof of \cite[Proposition 2.2]{CohKum2007}.

\begin{lemma}\label{lem:H}
Let $a$ be differentiable up to $K$ on a subset of $[-1, 1)$, and $g_1, g_2$ be two polynomials such that $\deg(g_1)+\deg(g_2)\le K$, then
$H(a,g_1g_2)=H(a,g_1)+g_1H(Q(a,g_1),g_2)$ where 
$$Q(a,g):=\frac{a-H(a,g)}{g}.$$

\end{lemma}
We provide a variation of \cite[Proposition 2.2]{CohKum2007} below. The proof is also similar.

\begin{proposition}\label{pro:2.2}
Let $c, d \in \R$. If $a$ is (strictly) $K$-absolutely monotone on $(c,d)$, then given any nonconstant polynomial $g$,
$Q(a,g)=\frac{a-H(a,g)}{g}$
is (strictly) absolutely monotone up to $K-\deg g$ on $(c,d)$.
\end{proposition}
\begin{proof}
By \cite[Lemma 2.1]{CohKum2007}, 
\begin{equation}\label{equ:2.1}
Q(a,g)(t)=\frac{a(t)-H(a,g)(t)}{g(t)}=\frac{a^{(\deg g)}(\xi)}{\deg g!}
\end{equation} for some $\xi\in(c,d)$. 

A direct consequence of Lemma \ref{lem:H} is that $Q(a, g_1g_2) = Q(Q(a, g_1), g_2)$.
For $n\in[1:K-\deg(g)], s_0\in(c,d)$, there exists $\xi'\in(c,d)$ such that
\begin{equation}\label{equ:Q}\frac{Q(a,g)^{(n)}(s_0)}{n!}=Q\left(Q(a,g),(t-s_0)^n\right)(s_0)=Q(a, (t-s_0)^ng)(s_0)=\frac{a^{(n+\deg g)}(\xi')}{(n+\deg g)!}.\end{equation}
The right hand side of \eqref{equ:Q} is nonnegative due to the absolute monotonicity of $a$.
\end{proof}

We also need to define a different version of conductivity here.
\begin{definition}
A nonconstant polynomial $g$ with all its roots in $[-1,1)$ is \emph{$K$-conductive} if for any $K$-absolutely monotone function $a$ on $[-1,1)$, $H(a,g)$ is positive definite.
\end{definition}

The following Lemma is a variation of \cite[Lemma 5.3]{CohKum2007}.

\begin{lemma}\label{lem:conductive}
If $g_1$ and $g_2$ are $K$-conductive and $g_1$ is positive definite, then $g_1g_2$ is ($K+\deg g_1$)-conductive.
\end{lemma}
\begin{proof}
Let $a$ be ($K+\deg g_1$)-absolutely monotone, then $Q(a,g_1)$ is $K$-absolutely monotone according to Proposition \ref{pro:2.2}. Consequently, $H(Q(a,g_1), g_2)$ is positive definite due to the conductivity of $g_2$. Finally,
$H(a,g_1g_2)=H(a,g_1)+g_1H(Q(a,g_1),g_2)$ is positive definite because all three functions are positive definite and positive definite functions are closed under taking products.
\end{proof}

\begin{proof}[Proof of Theorem \ref{thm:CK}]
Let $-1\leq t_1<t_2<\cdots<t_m<1$ be the $m$ distinct inner products of the $m$-sharp configuration.

Let $a(t)=f(2-2t)$ be defined on $[-1,1)$ and $h(t)$ be the Hermite interpolating polynomial that agrees with $a(t)$ to order 2 at each $t_i$ (i.e. $h(t_i)=a(t_i)$ and $h'(t_i)=a'(t_i)$). Then using our notation, $h=H(a, F^2)$ where $F=\prod_{i=1}^m(t-t_i)$.

For $r\in[-1,1)$, $l(t)=t-r$ is $K$-conductive for any $K\geq0$ since $H(a,l)$ is the nonconstant polynomial $a(r)$. It is also proven in \cite[Section 5]{CohKum2007} that $\prod_{i=1}^j(t-t_i)$ is strictly positive definite for all $j\leq m$.

For any $K\geq0$, $g_1=t-t_1, g_2=t-t_2$ are both $K$-conductive and $g_1$ is positive definite, then Lemma \ref{lem:conductive} implies that $g_1g_2$ is ($K+1$)-conductive. Using Lemma \ref{lem:conductive} repeatedly on $g_1=t-t_j, g_2=\prod_{i=1}^{j-1}(t-t_i)$, we get that $F^2$ is $K$-conductive for any $K\geq 2m$. In particular $F^2$ is $2m$-conductive and it follows that $h=H(a,F^2)$ is positive definite.

It is also clear that $h(t)\leq a(t)$ by applying \eqref{equ:2.1} with $g=F^2$. By \cite[Proposition 4.1]{CohKum2007}, the energy has a lower bound that is achieved by the $m$-sharp configuration.

If further $f$ is strictly $2m$-completely monotone, the uniqueness is the same as in \cite[Section 6]{CohKum2007} where only $a^{(\deg h+1)}(t)>0$ is needed. This is true since $\deg h+1\leq 2m$.

\end{proof}

\section*{Acknowledgements} This work was partially supported by the NSF REU grant DMS1359307 at the University of Maryland, College Park. The authors are thankful to Alan Bangura, Michael Dworken, Meghana Raja, Rosemary Smith, Hetian Wu, and Ran Zhang who have partially worked on the project as participants to the MAPS-REU at the University of Maryland. They would also like to thank, 
Henry Cohn, Dmitriy Bilyk,   Alexey Glazyrin, Matt Guay, Paul Koprowski, Dustin Mixon, Shayne Waldron, and Wei-Hsuan Yu for fruitful discussions. 
K. A. O.\ was partially supported by a grant from the Simons Foundation $\# 319197$,  the U. S.\ Army Research Office  grant  W911NF1610008,  the National Science Foundation grant DMS 1814253, and an MLK  visiting professorship. 

\bibliographystyle{amsplain}
\bibliography{pfpotential}

\end{document}